\documentclass[amsthm]{amsart}

\usepackage{amsmath, amsfonts,amssymb,graphicx,amsthm}
\usepackage[usenames]{color}

\newtheorem{thm}{Theorem}[section]
\newtheorem{defn}[thm]{Definition}
\newtheorem{cor}[thm]{Corollary}

\newtheorem{lem}[thm]{Lemma}
\newtheorem{prop}[thm]{Proposition}
\newtheorem{assu}[thm]{Assumption}

\newtheorem{alg}[thm]{Algorithm}

\newtheorem{rem}[thm]{Remark}

 \newcommand{\A}{\mathcal{A}}
 \newcommand{\I}{{\mathcal I}}
 \newcommand{\J}{{\mathcal J}}
 \newcommand{\M}{\mathcal{M}}

 \newcommand{\Rc}{\mathcal{R}}
 
 \newcommand{\C}{\mathbb{K}}
 \newcommand{\N}{\mathbb{N}}

 \newcommand{\GBs}{Gr{\"o}bner Bases}
 
 \newcommand{\BEZ}{\mathfrak{B}}
 \newcommand{\BD}{\mathcal{B}}
 \newcommand{\RBEZ}{\overline{\mathfrak{B}}}
 \newcommand{\RBD}{\overline{\mathcal{B}}}
 \def\xb{{\bf x}}
 \def\yb{{\bf y}}
 \def\im{{\mathrm{im}}}
\def\CC{\mathbb{K}}
\def\KK{\mathbb{K}}
\def\fh{{\bf f}^h}
\def\f{{\bf f}}
\def\Mac{{\rm Mac}_{\Delta}(\mathbf{f})}

\newcommand{\At}{\mathtt{A}}
\newcommand{\Bt}{\mathtt{B}}
\newcommand{\Sbs}{B}
\def\ead{\email}
\newenvironment{keyword}{\textbf{Keyword:}}{}
\begin{document}

\title{On the Computation of Matrices of Traces and  Radicals of Ideals}

\thanks{This research was partly supported by the Marie-Curie Initial
Training Network SAGA}

\author[I. Janovitz-Freireich]{Itnuit Janovitz-Freireich}
\address{Itnuit Janovitz-Freireich, Departamento de Matem\'aticas,
Centro de Investigaci\'on y Estudios Avanzados del I.P.N., Mexico City, Mexico}
\ead{janovitz@math.cinvestav.mx}

\author[B. Mourrain]{Bernard Mourrain}
\address{Bernard Mourrain, GALAAD, INRIA, Sophia Antipolis, France}
\ead{mourrain@sophia.inria.fr}

\author[L. R\'{o}nyai]{Lajos R\'{o}nyai}
\address{Lajos R\'{o}nyai, Computer and Automation Institute, 
Hungarian Academy of Sciences and Budapest University of Technology
and Economics, Budapest, Hungary}
\ead{lajos@csillag.ilab.sztaki.hu}

\author[A. Sz\'{a}nt\'{o}]{\'{A}gnes Sz\'{a}nt\'{o}}
\address{\'{A}gnes Sz\'{a}nt\'{o}, Mathematics Department, North Carolina State University, Raleigh, NC, USA}
\ead{aszanto@ncsu.edu}

\maketitle
\begin{abstract}
Let $f_1,\ldots,f_s \in \mathbb{K}[x_1,\ldots,x_m]$ be a system of
polynomials  generating a zero-dimensional ideal $\I$, where
$\mathbb{K}$ is an arbitrary  algebraically closed field. We study the computation of  ``matrices of traces" for the factor algebra $\A := \CC[x_1, \ldots , x_m]/ 
\I$, i.e. matrices with entries which are trace functions of the roots of $\I$. Such matrices of traces in turn allow 
us to compute a system of multiplication matrices $\{M_{x_i}|i=1,\ldots,m\}$ of the radical
$\sqrt{\I}$.

We first propose a method using Macaulay type resultant matrices of  $f_1,\ldots,f_s$ and a polynomial $J$ to compute
moment matrices, and in particular matrices of traces for $\A$. Here $J$ is a
polynomial generalizing the Jacobian. We prove bounds on the degrees needed for the Macaulay matrix in the case when $\I$ 
has finitely many projective roots in $\mathbb{P}^m_\CC$. We also extend previous results which work only for the case where $\A$ is Gorenstein to the non-Gorenstein case. 

The second proposed method uses Bezoutian matrices to compute matrices of
traces of $\A$. Here we need the assumption that $s=m$ and $f_1,\ldots,f_m$
define an affine complete intersection. This second method also works if we
have higher dimensional components at infinity. A new explicit description of
the generators of $\sqrt{\I}$ are given in terms of Bezoutians.
\end{abstract}



\section{Introduction}

 This paper is a continuation
of our previous investigation in \cite{JaRoSza06, JaRoSza07} to
compute the approximate radical of a zero dimensional ideal which
has zero clusters. It turns out that the computationally most expensive part of the
method in \cite{JaRoSza06, JaRoSza07} is the computation of the
matrix of traces. We address this problem in the present paper.   Some of the results of this paper 
also appeared earlier in \cite{JaMoRoSza08}, however here we present generalized versions of those results and also add new results, as described below.

The computation of  the radical of a zero
dimensional ideal is a very important problem in computer algebra
since a lot of the algorithms for solving polynomial systems with
finitely many solutions need to start with a radical ideal.  This is
also the case in many numerical approaches, where Newton-like
methods are used. From a symbolic-numeric perspective, when we are
dealing with approximate polynomials, the zero-clusters create great
numerical instability, which can be eliminated by computing the
approximate radical.

The theoretical basis of the symbolic-numeric algorithm presented in
\cite{JaRoSza06, JaRoSza07} was Dickson's lemma \cite{dickson},
which, in the exact case, reduces the problem of computing the
radical of a zero dimensional ideal to the computation of the
nullspace of the so called matrices of traces (see Definition
\ref{mtraces}):  in \cite{JaRoSza06, JaRoSza07} we studied numerical
properties of the matrix of traces when the roots are not multiple
roots, but form small clusters. Among other things we showed that
the direct computation of the matrix of traces (without the
computation of the multiplication matrices) is preferable since the
matrix of traces is continuous with respect to root perturbations
around multiplicities while multiplication matrices are generally
not.

In the present paper, first  we give a simple algorithm using only Macaulay  type
resultant matrices and elementary linear algebra to compute matrices
of traces of zero dimensional ideals which have finitely many projective roots. 
We also extend the method presented in \cite{JaMoRoSza08}  to handle systems which might 
have roots at infinity or for which the quotient algebra is non-Gorenstein. 

In the second part of the paper, we investigate how to compute matrices of traces using
Bezoutians in the affine complete
intersection case. Our approach in that case is based on \cite{MouPa00, mb-bqrs-jsc-05}.

For the method using Macaulay matrices we need the following assumptions: let
  ${\bf f}=[f_1,\ldots,f_s]$ be a system of polynomials of degrees
  $d_1\geq\cdots\geq d_s$ in $\CC[{\bf x}]$, with ${\bf
x}=[x_1,\ldots,x_{m}]$, generating an ideal $\I$ in $\CC[\xb]$,
where $\CC$ is an arbitrary  algebraically closed field. We assume
that the algebra  $\A:= \CC[{\bf x}]/\I$ is finite dimensional over
$\CC$ and that we have bounds $\delta>0$ and $0\leq k\leq \delta$ such that a basis
$\Sbs=[b_1, \ldots, b_N]$ of $\A$ can be  obtained
 by taking a linear basis of the space 
 $$
\mathbb{K}[\xb]_{k}/\langle f_1, \ldots,
f_s\rangle_{\delta}\cap \mathbb{K}[\xb]_{k}
$$
where  $\mathbb{K}[\xb]_{k}$ is the set of polynomials of degree at most $k$ 
and $\langle f_1, \ldots,
f_s\rangle_{\delta}=\{\sum_{i=1}^s q_if_i\;:\; \deg q_i\leq \delta-d_i\}$. We can
assume that the basis $\Sbs$ consists of monomials of degrees at most
$k$  by a slight abuse of notation.  In our earlier work \cite{JaMoRoSza08},
we gave bounds for $k$ and $\delta$ in the case where there were no roots at
infinity using a result of Lazard \cite{La81} (see Theorem
\ref{deltatheorem1}). Here we extend those results to the case where $\I$
has finitely many {\it projective} common roots 
in $\mathbb{P}^m_\CC$  (see Theorem
\ref{deltatheorem}). Furthermore, we now extend the method presented in
\cite{JaMoRoSza08}, which only addressed the case where  $\A$ is Gorenstein
over $\CC$ (see Definition \ref{gorenstein}), to handle non-Gorenstein
algebras. 

The main ingredient of our first method is a Macaulay type resultant
matrix $\Mac$, which is defined from the transpose matrix  of the degree $\Delta$ Sylvester
map $(g_1, \ldots, g_s)\mapsto \sum_{i=1}^s f_ig_i\in
\CC[\xb]_\Delta$ for $\Delta\leq 2\delta+1$ using simple linear algebra (see Definition
\ref{bigdelta}). Using our results,  we can compute a
basis $\Sbs$ of $\A$ using $\Mac$.  We also prove that a random element $\yb$ of
the nullspace of $\Mac$ provides an $N\times N$ moment matrix ${\mathfrak
M}_\Sbs(\yb)$ which has the maximal possible rank with high probability
(similarly as in \cite{LaLaRo07}).  Note that in the Gorenstein case the
moment matrix ${\mathfrak M}_\Sbs(\yb)$ is non-singular. This will no longer be
true in the non-Gorenstein case.  This moment matrix allows us to compute the
other main ingredient of our algorithm, a polynomial $J$ of degree at most
$\delta$, such that $J$ is the generalization of the Jacobian of $f_1, \ldots f_s$
in the case when $s=m$. The main result of the paper now can be formulated as
follows: \\

{\sc Theorem} {\it Let $\Sbs=[b_1, \ldots, b_N]$ be a basis of $\A$
with $\deg(b_i)\leq k$. With $J$ as above, let ${\rm Syl}_\Sbs(J)$ be
the transpose matrix  of the map $\sum_{i=1}^N c_i b_i\mapsto
J\cdot\sum_{i=1}^N c_i b_i\;\in \CC[x]_\Delta$ for $c_i\in \CC$.
Then $$ \left[Tr(b_ib_j)\right]_{i,j=1}^N= {\rm Syl}_\Sbs(J)\cdot X,
$$
where $X$ is the unique extension of the matrix $\mathfrak{M}_\Sbs({\bf
y})$ such that
$\Mac \cdot X=0.$}\\

Once we compute the matrix of traces
$R:=\left[Tr(b_ib_j)\right]_{i,j=1}^N$ and the matrices
$R_{x_k}:=\left[Tr(x_kb_ib_j)\right]_{i,j=1}^N= {\rm
Syl}_\Sbs(x_kJ)\cdot X$ for $k=1, \ldots, m$, we can use the results of
\cite{JaRoSza06, JaRoSza07} to compute a system of multiplication
matrices for the (approximate) radical of $\I$ as follows: if
$\tilde{R}$ is a (numerical) maximal non-singular submatrix of $R$
and $\tilde{R}_{x_k}$ is the submatrix of $R_{x_k}$ with the same
row and column indices as in $\tilde{R}$, then  the solution
$M_{x_k}$ of  the linear matrix equation
$$\tilde{R}M_{x_k}=\tilde{R}_{x_k}$$
is an (approximate) multiplication matrix of $x_k$  for the
(approximate) radical of $\I$.
See \cite{JaRoSza07} for the definition of (approximate)
multiplication matrices. Note that a generating set for the radical
$\sqrt{\I}$ can be obtained directly from the definition of
multiplication matrices, in particular, it corresponds to the rows
of the matrices $M_{x_1}, \ldots, M_{x_m}$.

We also point out that in the $s=m$ case these multiplication
matrices $M_{x_k}$ of $\sqrt{\I}$ can be obtained even more simply using the
nullspace of $\Mac$ and the Jacobian $J$ of ${\bf f}$, without
computing the matrices of traces.

In the last section we investigate the use of Bezoutians to compute matrices of 
traces of systems $f_1,\ldots,f_m$ which form an affine complete
intersection. In this particular setting, our method allows  
systems that may have higher dimensional projective components.

In the univariate case it is proved in \cite{MouPa00} that the Bezoutian
matrix of a univariate polynomial $f$ and its derivative $f'$ is
a matrix of traces with respect to the Horner basis of $f$ (see subsection
\ref{unisection}). Therefore, applying the method to compute the approximate
or exact radical from the matrix of traces provided by the Bezoutian will
give us an approximate or exact square-free factorization of $f$. The
question that naturally arises is how this method relates to computing the
square-free factor as $\frac{f}{gcd(f,f')}$. We show here that the two
algorithms are computationally equivalent.

The generalization to the multivariate case is not quite as
straightforward. The goal would be to express the Bezout matrix of
$f_1,\ldots,f_m$ and their Jacobian $J$ as a matrix of traces with
respect to some basis, generalizing the univariate case (see the
definition of the Bezout matrix -- sometimes also referred as the
Dixon matrix -- in Definition \ref{bezmatrixdef}). Unfortunately,
the Bezout matrix cannot directly be expressed as a matrix of
traces. However, in \cite{MouPa00} it is shown that a reduced
version of the Bezout matrix of $f_1,\ldots,f_m,$ and $J$ is equal
to the matrix of traces of $f_1,\ldots,f_m$ with respect to the so
called canonical basis, obtained from the reduced Bezout matrix of
$f_1,\ldots,f_m$, and $1$. The required reduction of the Bezout
matrix involves reducing polynomials modulo $\I$.

Now the question is how to find the reduced version of the Bezout
matrix without further information on the structure of the
quotient algebra $\mathbb{C}[x_1, \ldots, x_m]/\I$, e.g. without
Gr\"obner Bases or multiplication matrices. First we show that we can obtain a 
set of generating polynomials for the radical $\sqrt{I}$ from the {\em non-reduced} Bezoutian matrices (see Theorem \ref{radbez}). 
Secondly, we give an algorithm which computes a system of multiplication matrices 
$M_{x_1}, \ldots, M_{x_m}$ for $\sqrt{I}$. This algorithm  adapts the
results of \cite{mb-bqrs-jsc-05}  to find the required reduced  Bezout matrices using only elements in 
$\sqrt{I}$ which were obtained from non-reduced Bezout matrices. 

\section{Related Work}

The motivation for this work was the papers \cite{LaLaRo07,
LaLaRo072} where they use moment matrices to compute the radical of
real and complex ideals. They present two versions of the method for
the complex case: first, in \cite{LaLaRo072} they double up the
machinery for the real case to obtain the radical of the complex
ideal.   However, in \cite{LaLaRo07} they significantly simplify
their method and show how to use moment matrices of maximal rank to
compute the multiplication matrices of an ideal between $\I$ and its
radical $\sqrt{\I}$. In particular, in the Gorenstein case they can
compute the multiplication matrices of $\I$. In fact, in
\cite{LaLaRo07} they cite our previous work \cite{JaRoSza06} to
compute the multiplication matrices of $\sqrt{\I}$  from  the
multiplication matrices of $\I$, but the method proposed in the
present paper is much simpler and more direct.

Note that one can also obtain the multiplication matrices  of $\I$
with respect to the basis $\Sbs=[b_1, \ldots, b_N]$ by simply
eliminating the terms not in $\Sbs$ from $x_kb_i$ using ${\rm
Mac}_{\delta+1}({\bf f})$. The advantage of computing multiplication
matrices of the radical $\sqrt{\I}$  is that it  returns matrices
which are always simultaneously diagonalizable, and possibly smaller
than the multiplication matrices of $\I$, hence easier to work with.
Moreover, if $\Sbs$ contains the monomials $1,x_{1},\ldots,x_{m}$, one
eigenvector computation yields directly the coordinates of the
roots.

Computation of  the radical of zero dimensional complex ideals is
very well studied in the literature:  methods most related to ours
include  \cite{GonVeg, becwor96} where matrices of traces are used
in order to find  generators of the radical, and the matrices of
traces are computed using \GBs; also, in \cite{ArSo95} they use the
traces to give a bound for the degree of the generators of the
radical and use linear solving methods from there; in \cite{GoTr95}
they describe the computation of the radical using symmetric
functions which are related to traces. One of the most commonly
quoted method to compute radicals is to compute the  projections
$\I\cap\CC[x_i]$ for each $i=1, \ldots, m$ and then use univariate
squarefree factorization (see for example
\cite{GiTrZa88,KrLo91b,Cox98,GrPf02} ). The advantage of the latter
is that it can be generalized for higher dimensional ideals (see for
example \cite{KrLo91}). We note here that an advantage of the method
using matrices of traces is that it behaves stably under
perturbation of the  roots of the input system, as was proved in
\cite{JaRoSza07}. Other methods to compute the radical of zero
dimensional ideals include \cite{KoMoHo89,GiMo89,
Lak90,Lak91,LakLaz91,YoNoTa92}. Applications of computing the
radical include \cite{HeObPa06}, where they show how to compute the
multiplicity structure of the roots of $\I$ once the radical is
computed.

Methods for computing the matrix of traces directly from the
generating polynomials of $\I$, without using multiplication
matrices, include \cite{DiazGonz01,brigonz02} where they use Newton
Sums, \cite{carmou96,CatDicStu96,CatDicStu98} where they use
residues and \cite{DaJe05} using resultants. Besides computing the
radical of an ideal,  matrices of traces have numerous applications
mainly in real algebraic geometry \cite{Be91, PeRoSz93, BeWo94}, or
in \cite{Rouiller99}  where trace matrices are applied  to find
separating linear forms deterministically.

\section{Ideals with Finitely Many Projective Roots}
\subsection{The Gorenstein Case}

Some of the results of this subsection appeared in \cite{JaMoRoSza08}. We
included them here for completeness.\\ 

Let  ${\bf f}=[f_1,\ldots,f_s]$ be a system of polynomials of
degrees $d_1\geq\cdots\geq d_s$ in $\mathbb{K}[{\bf x}]$, where
${\bf x}=[x_1,\ldots,x_{m}]$ and $\CC$ is an arbitrary algebraically
closed field. Let $\I$ be the ideal generated by $f_1,\ldots,f_s$ in
$\mathbb{K}[{\bf x}]$ and define $\A:=\mathbb{K}[{\bf x}]/\I$.  We
assume throughout the paper that $\A$  is a finite dimensional
vector space over $\CC$  and let $\A^*$ denote the dual space of
$\A$.

Let us first recall the definition of a Gorenstein algebra (c.f.
\cite{Kun86,ScSt75,ElMo07,LaLaRo07}). Note that these algebras are
also referred to as Frobenius in the literature, see for example
\cite{BeCaRoSz96}.

\begin{defn}\label{gorenstein}
A finite dimensional $\CC$-algebra $\A$ is Gorenstein (over $\CC$)
if there exists a nondegenerate $\CC$-bilinear form $B(x,y)$ on $\A$
such that
\[
B(ab,c)=B(a,bc) \; \text{ for every } a,b,c \in \A.
\]
\end{defn}

Note that this is equivalent to the fact that $\A$ and $\A^*$ are
isomorphic as $\A$ modules. It is also equivalent to the existence
of a ${\mathbb K}$-linear function $\Lambda: \A\rightarrow {\mathbb
K}$ such that the bilinear form $B(a,b):=\Lambda(ab)$ is
nondegenerate on $\A$.


\begin{assu}\label{assump}
Throughout this subsection we assume that $\A$ is Gorenstein.
Furthermore, we also assume that we have a bound $\delta>0$  and
$0\leq k\leq\delta$ such that

\begin{eqnarray}\label{deltacond}
\dim_\mathbb{K}\mathbb{K}[\xb]_{k}/\langle f_1, \ldots,
f_s\rangle_{\delta}\cap \mathbb{K}[\xb]_{k}
=\dim_\mathbb{K}\mathbb{K}[\xb]_{k}/\langle f_1, \ldots,
f_s\rangle_{d}\cap\mathbb{K}[\xb]_{k}
\end{eqnarray} for all $d\geq \delta$.
 Here $\mathbb{K}[\xb]_{k}:=\left\{p\in \mathbb{K}[\xb]\;:\;\deg(p)\leq k\right\}$ and
\begin{eqnarray}\label{fd}
\langle f_1, \ldots, f_s\rangle_{d}:=\left\{\sum_i f_i q_i\; :\;
   \deg(q_i) \leq d-d_i\right\}.
 \end{eqnarray}
 
\end{assu}

\begin{thm}
Assume that $\delta$ and $k$ satisfy the condition
(\ref{deltacond}). Then
\[
\dim_\mathbb{K}(\A)=\dim_\mathbb{K}\mathbb{K}[\xb]_{k}/\langle f_1,
\ldots, f_s\rangle_{\delta}\cap \mathbb{K}[\xb]_{k}.
\]
\end{thm}

\begin{proof}
Assume that $\delta$ and $k$ satisfy the condition (\ref{deltacond})
and let $B:=[b_1, \ldots, b_N]$ be a basis for
$\mathbb{K}[\xb]_{k}/\langle f_1, \ldots, f_s\rangle_{\delta}\cap
\mathbb{K}[\xb]_{k}$. Taking pre-images, we can assume that $b_1,
\ldots, b_N$ are polynomials in $\mathbb{K}[\xb]_{k}$. We claim that
$B$ is a basis for $\A=\mathbb{K}[\xb]/\langle f_1, \ldots,
f_s\rangle$. Since $ \langle f_1, \ldots,
f_s\rangle_{\delta}\subseteq \langle f_1, \ldots, f_s\rangle_{d}$ if
$\delta\leq d$, $B$ is clearly a generator set for $\A$. On the
other hand, assume that $B$ is not linearly  independent in $\A$,
i.e. there exist $c_1, \ldots, c_N\in \CC$ such that $\sum_{i=1}^N
c_ib_i$ is in $ \langle f_1, \ldots, f_s\rangle$. Then there exists
$d\geq \delta$ such that   $\sum_{i=1}^N c_ib_i\in \langle f_1,
\ldots, f_s\rangle_d$. But $\sum_{i=1}^N c_ib_i$ is also in $
\mathbb{K}[\xb]_{k}$, so $B$ is linearly dependent in
$\mathbb{K}[\xb]_{k}/\langle f_1, \ldots, f_s\rangle_{d}\cap
\mathbb{K}[\xb]_{k}$, which contradicts condition (\ref{deltacond}).
\end{proof}

We have the following theorems giving bounds for $\delta$ in the
case when ${\bf f}$ has finitely many projective roots. First we
assume that ${\bf f}$ has no roots at infinity.

\begin{thm}\label{deltatheorem1}
Let  ${\bf f}=[f_1,\ldots,f_s]$ be a system of polynomials of
degrees $d_1\geq\cdots\geq d_s$ in $\mathbb{K}[{\bf x}]$. Assume
that the corresponding system of homogenous polynomials
$f^h_1,\ldots,f^h_s$ has finitely many projective common roots in
$\mathbb{P}_\CC^m$. Assume further that $f_1, \ldots ,f_s$ have
no common roots at infinity. Then:
\begin{enumerate}

\item If $s=m$  then for $\delta=k:=\sum_{i=1}^m (d_i-1)$
    condition (\ref{deltacond}) is satisfied. Furthermore, in this case $\A$ is always Gorenstein.

\item If $s> m$ then
for $\delta=k:=\sum_{i=1}^{m+1} d_i-m$ condition (\ref{deltacond})
is satisfied.
\end{enumerate}
 \end{thm}

\begin{proof}

For the first assertion let $\fh=[f_1^h, \ldots, f_m^h]$ be the
homogenization  of ${\bf f}$ using a new variable $x_{m+1}$. Using
our assumption that $\fh$  has finitely many roots in
$\mathbb{P}_{\mathbb{K}}^m$ and $s=m$, one can see that $(\fh)$ is a
regular sequence in $R:={\mathbb K}[x_1, \ldots, x_m,x_{m+1}]$.
Define the graded ring  $B:=R/\langle \fh\rangle$. Following the
approach
and notation
 in
\cite{Stan96}, we can now calculate the Hilbert series of $B$,
defined by $H(B,\lambda)=\sum_d \mathcal{H}_B(d) \lambda^d$, where
${\mathcal H}_B$ is the Hilbert function of $B$. We have
$$H(R,\lambda)=\frac{H(B,\lambda)}{(1-\lambda^{d_1})\cdots (1-\lambda^{d_m})},$$
and using the simple fact that
$$H(R, \lambda)=\frac{1}{(1-\lambda)^{m+1}}$$
we obtain that
$$
\begin{aligned}
 H(B,\lambda)&=\frac{(1+\lambda+\cdots +\lambda^{d_1-1})\cdots
(1+\lambda+\cdots
+\lambda^{d_m-1})}{(1-\lambda)}\notag\\
&=g(\lambda)(1+\lambda+\ldots),\notag
 \end{aligned}$$
 where
$$g(\lambda) =(1+\lambda+\cdots +\lambda^{d_1-1})\cdots (1+\lambda+\cdots
+\lambda^{d_m-1}).$$ This implies that the Hilbert function
$${\mathcal H}_B(\delta)={\mathcal H}_B(\delta+1)={\mathcal H}_B(\delta+2)=\ldots $$
Note that dehomogenization induces a linear isomorphism $B_d
\rightarrow {\mathbb K}[{\bf x}]_d/\langle f_1, \ldots,
f_m\rangle_{d}$, where $B_d$ stands for the degree $d$ homogeneous
part of $B$. From this, using that there are no common roots at
infinity, we infer that for $d\geq \delta$ $\dim _{\mathbb
K}{\mathbb K}[{\bf x}]_d/\langle f_1, \ldots,
f_m\rangle_{d}=\dim_\CC \A = N$, which implies  (\ref{deltacond}).

Note that the common value $N={\mathcal H}_B(\delta)$ is the sum of the
coefficients of $g$, which is
$$g(1)=\prod_{i=1}^m d_i.$$

To prove that $\A$ is Gorenstein, we cite  \cite[Proposition 8.25,
p. 221]{ElMo07} where it is proved that if $f_{1},\ldots, f_{m}$ is
an affine complete intersection then the Bezoutian
$B_{1,f_{1},\ldots,f_{m}}$ defines an isomorphism between ${\A}^*$
and $\A$.

To prove the second assertion we note that \cite[Theorem 3.3]{La81}
implies that
$$
\dim _{\mathbb K}B_\delta=\dim _{\mathbb K}B_{\delta+1}=\ldots .
$$
From here we obtain (\ref{deltacond}) as in the Case 1.
\end{proof}

The following theorem generalizes the previous result for systems which may have
roots at infinity. 

\begin{thm}\label{deltatheorem}
Let  ${\bf f}=[f_1,\ldots,f_s]$ be a system of polynomials of
degrees $d_1\geq\cdots\geq d_s$ in $\mathbb{K}[{\bf x}]$.
Assume that the corresponding system of homogenous polynomials
$f^h_1,\ldots,f^h_s$ has finitely many projective common roots in
$\mathbb{P}_\CC^m$ and let $\J:=\langle f^h_1,\ldots,f^h_s \rangle$ be
the ideal they generate in $R:=\mathbb{K}[x_1, \ldots, x_n,x_{n+1}]$. Then:
\begin{enumerate}

\item If $s=m$  then for $k:=\sum_{i=1}^m (d_i-1)$ and
$\delta:=k+1$ condition (\ref{deltacond}) is satisfied.
\item If $s> m$ then for $k:=\sum_{i=1}^{m+1} d_i-m$  and
$\delta:=k+1$ condition
(\ref{deltacond}) is  satisfied.
\end{enumerate}
  \end{thm}

\begin{proof}
Assume that ${\bf f}$ has $N$ affine roots and $N'$ roots at
infinity, counted with multiplicity.
In this proof only, for a homogeneous ideal $J\subseteq R$, $J_t$ denotes the elements of $J$ of degree {\em equal} to $t$,  abusing the notation. By the proof of Theorem \ref{deltatheorem1} and \cite{La81}, we have 
that for
$k$ defined above (in both cases) and for all $d\geq 0$,
$R_{k+d}/\J_{k+d}$
(resp. $R_{k+d}/(\J+(x_{m+1}))_{k+d}$)
  is of dimension $N+N'$
(resp. $N'$).
Consider the exact sequence
$$
0\rightarrow (\J:x_{m+1})_{k}/\J_{k} \rightarrow R_{k}/\J_{k}
\stackrel{\mathcal{M}_{m+1}}{\longrightarrow} R_{k+1}/\J_{k+1}  \rightarrow
R_{k+1}/(\J+x_{m+1})_{k+1}\rightarrow 0
$$
where $\mathcal{M}_{m+1}$ is the multiplication by $x_{m+1}$.
Using the  relation on the dimensions of the vector
spaces of this exact sequence, we deduce that
$$
\dim_{\KK} ((\J:x_{m+1})_{k}/\J_{k})= (N+N')-(N+N')+ N' =N'.
$$
Thus we can choose a basis $[b_1^h,\ldots,b_{N+N'}^h]$ of $R_{k}/\J_{k}$
such that  $b_1^h,\ldots,b_{N'}^h \in  (\J:x_{m+1})_{k}$.
Moreover, we can even assume that
$b^{h}_{N'+1},\ldots, b^{h}_{N+N'} \in \langle x_{m+1} \rangle_{k}$,
since
$$
{\rm span}( b_1^h,\ldots,b_{N+N'}^h )
= {\rm span}(  b_1^h,\ldots,b_{N'}^h) +
\langle x_{m+1} \rangle _{k}\ +\  \J_{k}.
$$

If $x_{m+1}^{k} \in (\J:x_{m+1})_{k}$ then $x_{m+1}^{k+1}\in \J$,
all the roots are at infinity, $N=0$,
and $(\J:x_{m+1})_{k} = R_{k} \ \mathrm{modulo}\ \J_{k}$ which shows that
$x_{m+1} R_{k} \subset \J_{k+1}$. After dehomogeneization
all polynomials of degree $\le k$ are in $ \langle f_1, \ldots, f_s\rangle_{k+1}$.
So  condition (\ref{deltacond}) is satisfied for $\delta=k+1$.

Suppose now that $x_{m+1}^{k} \not\in (\J:x_{m+1})_{k}$, so that
we can take $b_{N'+1}^h=x_{m+1}^{k}$.


As $x_{m+1} (\J:x_{m+1}) \subset \J$, we deduce that
$x_{m+1}\, b^h_{1}=\cdots=x_{m+1}\,b^h_{N'}=0$ modulo $\J_{k+1}$
and that
$$
\dim_\KK \left( {\rm span}\left(  x_{m+1}\,b^h_{1},\ldots,
x_{m+1}\,b^h_{N+N'}\right)/  \J_{k+1}\right)\leq N.
$$

As we have
$$
R_k = {\rm span}\left(  b_1^{h}, \ldots ,b_{N+N'}^{h}\right)+  \J_k
$$
and $x_{m+1}\, b^h_{i}\in \J_{k+1}$ for $1\leq i\leq N'$, we deduce
that
$$
x_{m+1}\, R_k = {\rm span}\left( x_{m+1}b_{N'+1}^{h}, \ldots , 
x_{m+1}b_{N+N'}^{h}\right) \ + \  \J_{k+1}.
$$
After dehomogeneization, we obtain a family $B=[b_{N'+1}, ...,b_{N+N'}]$
of $N$ elements  of degree $< k$ (because $b_i^{h} \in \langle
x_{m+1}\rangle_{k}$ for $N'< i \leq N+N'$) such that
$$
\KK[x_{1},\ldots, x_{m}]_{k} = {\rm span}(  B ) +
\langle f_{1}, \ldots, f_s  \rangle_{k+1} \cap \KK[x_{1},\ldots, 
x_{m}]_{ k}
$$
(here we use the notation of Assumption \ref{assump} again). Thus any polynomial of degree $\leq k$ can be rewritten, modulo $ \langle
f_1,\ldots, f_{s} \rangle_{k+1}$, as a linear combination of elements in 
$B$ of degree $<k$.
As $B$ contains $1$ since $b_{N'+1}^h=x_{m+1}^{k}$, this shows that $B$ is a
generating set of $\A$.
As  $\A$ is of dimension $N$, $B$ is in fact a basis of $\A$, and thus 
$\delta:=k+1$
and $k$ satisfy the conditions in (\ref{deltacond}).
\end{proof}

\begin{rem} Note that in general 
$$ \langle f_1,\ldots,f_s \rangle\cap{\mathbb K}[{\bf x}]_d\ \neq \langle f_1,\ldots,f_s \rangle_d,$$ where  $\langle f_1,\ldots,f_s \rangle_d$ was defined in (\ref{fd}). Inequality can happen when the system has a root at infinity, for example, if $f_1=x+1, \;f_2=x$ then
$\langle f_1,f_2 \rangle\cap{\mathbb K}[{\bf x}]_0={\mathbb K}$ but $\langle f_1,f_2 \rangle_0=\{0\}$.
However, using the  homogenization $f_1^h,\ldots,f_s^h$, the degree
$d$ part of the homogenized ideal is always equal to the space
spanned by the multiples of $f_1^h,\ldots,f_s^h$ of  degree $d$. The
above example also demonstrates that $\dim \A$ is not always the
same as $\dim\CC[\xb]_d/\langle f_1,\ldots,f_s \rangle_d$ even for large $d$, because above $\dim\A=0$ but  $\dim\CC[x,y]_d/\langle
f_1,f_2 \rangle_d=1$ for all $d\geq 0$.
\end{rem}

\begin{defn}\label{D}
Let $N:=\dim_\mathbb{K}(\A)$ and  fix $\Sbs=[b_1,\ldots,b_N]$ a
monomial basis for $\A$ such that ${\rm deg}( b_i)\leq k$ for
all $i=1,\ldots,N$. We define $D$ to be the maximum degree of the monomials
in $\Sbs$. Thus $D\leq k\leq \delta$.
\end{defn}

Next we will define Sylvester and Macaulay type resultant matrices
for $f_1,\ldots f_s$.

\begin{defn}\label{bigdelta}
Define $$\Delta:=\max(\delta-1,2D)$$ where $\delta$ and $D$ are
defined in Assumption \ref{assump} and Definition \ref{D}.

Let ${\rm Syl}_{\Delta+1}(\mathbf{f})$ be the transpose matrix of the
linear map
\begin{eqnarray}
\bigoplus_i \CC[{\bf x}]_{\Delta-d_i+1}&\longrightarrow &\CC[{\bf x}]_{\Delta+1}\label{syl}\\
(g_1,\ldots,g_s)&\mapsto& \sum_{i=1}^sf_ig_i\notag
\end{eqnarray}
written in the monomial bases. So, in our notation, ${\rm
Syl}_{\Delta+1}(\mathbf{f})$ will have rows which correspond to all
polynomials $f_ix^\alpha$ of degree at most
   $\Delta$.

Let $\Mac$ be the matrix with rows corresponding to a basis of  $\langle f_1, \ldots, f_s\rangle_{\Delta+1}\cap  \CC[{\bf x}]_{\Delta}$, obtained by eliminating coefficients of terms of degree $\Delta+1$ in the matrix ${\rm
Syl}_{\Delta+1}(\mathbf{f})$ using Gaussian elimination, and then taking a maximal linearly independent set  among the eliminated rows. \end{defn}

\begin{rem} In the case where $s=m$, for generic ${\bf f}$ with no roots at infinity, we can directly construct $\Mac$ by taking the restriction of the map (\ref{syl}) to
\[
\bigoplus_{i=1}^m {\mathcal S}_i(\Delta)\longrightarrow \CC[{\bf
x}]_{\Delta}
\]
where ${\mathcal S}_i(\Delta)={\rm span}\{{\bf
x}^{\alpha}:|\alpha|\leq\Delta-d_i,\, \forall
 j< i,\,\alpha_j< d_j\}$.

Here $\Mac$ is a submatrix of the classical Macaulay matrix of the
homogenization of $\f$ and some $f^h_{m+1}$, where $f^h_{m+1}$ is
any homogeneous polynomial of degree $\Delta-\delta$: we only take
the rows corresponding to the polynomials in $\f$. Since the
Macaulay matrix is generically non-singular,\\ $\Mac$ will also be
generically full rank.

Note that with our assumption that $f_1,\ldots,f_m$ has no roots at infinity, we have that $\Mac$ has column \rm{corank} $
\dim \A=\prod_{i=1}^m d_i. $

\end{rem}

Since $\Delta\geq\delta-1$, by Assumption \ref{assump} and Theorem \ref{deltatheorem}, the corank of
$\Mac=N$, where $N$ is the dimension of $\A$. Also, we can assume
that the first columns of $\Mac$ correspond
to a  basis $\Sbs$ of $\A$.

Fix an element $${\bf y}=[y_{\alpha}:\alpha\in \N^m, \;|\alpha|\leq
\Delta]^T$$ of the nullspace ${\rm Null}(\Mac)$, i.e.
$\Mac\cdot\yb=0$.

\begin{defn}\label{moment}
Let $\Sbs= [b_{1},\ldots,b_{N}]$ be the basis of $\A$ as above, consisting of monomials of
degree at most $D$.
 Using ${\bf y}$ we can  define $\Lambda_\yb\in \A^*$ by
$ \Lambda_\yb (g):=\sum_{\xb^\alpha\in \Sbs} y_{\alpha}g_{\alpha}, $
where $g=\sum_{\xb^\alpha\in \Sbs}g_{\alpha}\xb^\alpha\in \A$. Note
that every $\Lambda\in \A^*$ can be defined as $\Lambda_\yb$ for
some $\yb\in {\rm Null}(\Mac)$ or more generally with an element of
${\mathbb K}[{\bf x}]^{*}$ which vanishes on the ideal $\I$.

Define the \emph{moment matrix} $\mathfrak{M}_\Sbs({\bf y})$ to be the
$N\times N$ matrix given by
$$
 \mathfrak{M}_\Sbs({\bf y})=[y_{\alpha+\beta}]_{\alpha,\beta},$$
 where $\alpha$ and $\beta$ run through the exponents of the monomials in $\Sbs$. Note that $\mathfrak{M}_\Sbs$ is only a
   submatrix of the usual notion of moment matrix, see for example
   \cite{CuFi96}.

For $p \in \A$, we define the linear function  $p\cdot \Lambda\in
\A^*$ as $p\cdot \Lambda(g):=\Lambda(pg)$ for all $g\in \A$.

\end{defn}

\begin{rem}
If one considers a linear function $\Lambda$ on $\A$, such that the
bilinear form $(x,y)\mapsto\Lambda(xy)$ is nondegenerate on $\A$,
then the moment matrix corresponding to this $\Lambda$ will be the
one whose $(i,j)$-th entry  is just $\Lambda(b_ib_j)$. Moreover, for
$g,h\in \A$
\[
\Lambda_\yb(gh)={\rm coeff}_\Sbs(g)^T\cdot\mathfrak{M}_\Sbs(\yb)\cdot{\rm
coeff}_\Sbs(h)
\]
where ${\rm coeff}_\Sbs(p)$ denotes the vector of coefficients of $p\in
\A$ in the basis $\Sbs$.
\end{rem}

The following proposition is a simple corollary of \cite[Prop 3.3
   and Cor. 3.1]{LaLaRo07}.

\begin{prop}
Let $\,{\bf y}$ be a random element of the vector space ${\rm
Null}(\Mac)$. With high probability, $\mathfrak{M}_\Sbs({\bf y})$ is
non-singular.
\end{prop}


 \begin{rem}\label{rem0}
Using the above proposition, one can detect whether the algebra $\A$
is not Gorenstein with high probability by simply computing the rank
of $\mathfrak{M}_\Sbs({\bf y})$ for (perhaps several) random elements
$\,{\bf y}$ in ${\rm Null}(\Mac)$.
\end{rem}

\begin{rem}\label{rem1}
%
%
By \cite[Theorem 2.6 and Lemma 3.2]{LaLaRo07} one can extend ${\bf
y}$ to $\tilde{\bf y}\in\CC^{\N^m}$ such that the infinite moment
matrix $\mathfrak{M}(\tilde{\bf y}):=[\tilde{
y}_{\alpha+\beta}]_{\alpha, \beta\in \N^m}$ has the same rank as
$\mathfrak{M}_\Sbs({\bf y})$ and the columns of
$\mathfrak{M}(\tilde{\bf y})$ vanish on all the elements of the
ideal $\I$.
\end{rem}

Next we define a basis dual to $\Sbs=[b_1,\ldots,b_N]$ with respect to
the moment matrix $\mathfrak{M}_\Sbs({\bf y})$.  Using this dual basis
we also define a polynomial $J$ which is in some sense a
generalization of the Jacobian of a well-constrained polynomial
system.

\begin{defn} \label{dualbasis}

From now on we fix $\yb\in {\rm Null}(\Mac)$ such that
$\mathfrak{M}_\Sbs({\bf y})$ is invertible and we will denote by
$\Lambda$ the corresponding element $\Lambda_\yb\in \A^*$. We define

$$\mathfrak{M}_\Sbs^{-1}({\bf y})=:[c_{ij}]_{i,j=1}^N.$$

Let $b_i^*:=\sum_{j=1}^N c_{ji}b_j.$ Then $[b_1^*, \ldots, b_N^*]$
corresponds to the columns of the inverse matrix
$\mathfrak{M}_\Sbs^{-1}({\bf y})$ and   they also form a basis for
$\A$. Note that we have $\Lambda(b_ib_j^*)=1$, if $i=j$, and 0
otherwise.

Define the {\em generalized Jacobian} by
\begin{eqnarray}
J:=\sum_{i=1}^Nb_ib_i^*\;\text{ mod }  \I\label{J}
\end{eqnarray}
expressed in the basis $\Sbs=[b_1, \ldots, b_N]$ of $\A$.

\end{defn}

\begin{rem}
Note that since $\sum_{i=1}^N b_ib^*_i$ has degree at most $2D$, and
$\Delta\geq 2D$, we can use $\Mac$ to find its reduced form, which is
$J$. Because of this reduction, we have that ${\rm deg}(J)\leq D\leq
\delta$.

Note that the notion of generalized Jacobian   was also
introduced in \cite{BeCaRoSz96}. Its name come from the fact that if
$\,s=m$ and if $\,\Lambda$ is the so called residue (c.f.
\cite{ElMo07}),  then $\sum_{i=1}^N b_ib^*_i=J$ is the Jacobian of
$f_1, \ldots, f_m$.
\end{rem}

We now recall the definition of the multiplication matrices and the
matrix of traces as presented in \cite{JaRoSza07}.

\begin{defn}\label{mtraces}
Let $p\in \A$. The multiplication matrix $M_p$ is the transpose of
the matrix of the multiplication map
\[
\begin{aligned}
\mathcal{M}_p:\A&\longrightarrow \A\notag\\
g&\mapsto pg
\end{aligned}
\]
written in the basis $\Sbs$.

The \emph{matrix of traces} is the $N\times N$ symmetric matrix:
\[
T=\left[Tr(b_ib_j)\right]_{i,j=1}^N
\]
where $Tr(pq):= Tr(M_{pq})$,  $M_{pq}$ is the multiplication matrix
of $pq$ as an element in $\A$ in terms of the basis $\Sbs=[b_1, \ldots,
b_N]$ and $Tr$ indicates the trace of a matrix.
\end{defn}

The next results relate the multiplication by $J$ matrix to the
matrix of traces $T$.

\begin{prop}
Let $M_J$ be the multiplication matrix of $J$ with respect to the
basis $\Sbs$. We then have that
$$M_J=[Tr(b_ib_j^*)]_{i,j=1}^N. $$
\end{prop}

\begin{proof}
Let $\Lambda\in \A^*$ be as in Definition \ref{dualbasis}. For any
$h\in \A$ we have that
\[
\begin{aligned}
&h=\sum_{j=1}^N \Lambda(hb_j)b_j^*=\sum_{j=1}^N
\Lambda(hb_j^*)b_j\notag\\
\Rightarrow\quad &h\,b_i=\sum_{j=1}^N \Lambda(h\,b_i\,b_j^*)b_j\Rightarrow M_h[j,i]=\Lambda(h\,b_{i}\,b_j^*)\notag\\
\Rightarrow\quad &Tr(h)=\sum_{i=1}^N
\Lambda(h\,b_{i}\,b_i^*)=\Lambda(h\sum_{i=1}^N b_{i}\, b_i^*).\notag
\end{aligned}
\]

Since $J=\sum_{i=1}^N b_i^*b_i$ in $\A$, we have
$Tr(h)=\Lambda(hJ)$.
Therefore
\[
M_J[j,i]=\Lambda(J\, b_{i}\,b_j^*)=Tr(b_{i}\,b_j^*) \]
\end{proof}

\begin{cor}
\[
M_J\cdot\mathfrak{M}_\Sbs({\bf y})=[Tr(b_ib_j)]_{i,j=1}^N=T,
\]
or equivalently $J\cdot \Lambda = Tr$ in $\A^{*}$.
\end{cor}

\begin{proof}
The coefficients of $b_i^*$ in the basis $\Sbs=[b_1,\ldots,b_N]$ are
the columns of $\mathfrak{M}_\Sbs^{-1}({\bf y})$, which implies that
\[
M_J=[Tr(b_ib_j^*)]_{i,j=1}^N=[Tr(b_ib_j)]_{i,j=1}^N\cdot\mathfrak{M}_\Sbs^{-1}({\bf
y}).
\]
Therefore we have that $M_J\cdot\mathfrak{M}_\Sbs({\bf
y})=[Tr(b_ib_j)]_{i,j=1}^N$.
\end{proof}

Finally, we prove that the matrix of traces $T$ can be computed
directly from the Macaulay matrix of $f_1, \ldots, f_s$ and $J$,
without using the multiplication matrix $M_J$. First we need a
lemma.

\begin{lem}\label{Rlemma}
There exists a unique matrix $\mathfrak{R}_\Sbs({\bf y})$ of size
$|{\rm Mon}_{\leq}(\Delta)-\Sbs| \times |\Sbs|$ such that

\[
\Mac\cdot\begin{tabular}{|c|} \hline
\\
$\mathfrak{M}_\Sbs({\bf y})$\\
\\
\hline
\\
$\mathfrak{R}_\Sbs({\bf y})$\\
\\
\hline
\end{tabular}=0
\]

\end{lem}

\begin{proof}
By our assumption that the first columns of ${\rm
Mac}_{\Delta}(\f)$ correspond to $\Sbs$ we have
\[
\begin{tabular}{|ccc|}
\hline
&&\\
&$\Mac$&\\
&&\\
\hline
\end{tabular}=\begin{tabular}{|ccc|ccc|}
 \hline
&&&&&\\
&$\Bt$&&&$\At$&\\
&&&&&\\
\hline
\end{tabular},
\]
where the columns of $B$ are indexed by the monomials in $\Sbs$. Note
here that  by Definition \ref{bigdelta} and Assumption \ref{assump} the rows of $\Mac$ span
$\I_{\Delta+1}\cap \CC[\xb]_\Delta$, and the monomials in $\Sbs$ span the factor space
$\CC[\xb]_\Delta/(\I_{\Delta+1}\cap \CC[\xb]_\Delta)$. These together imply that the (square)
submatrix $\At$ is invertible.

Then
\[\begin{tabular}{|ccc|ccc|}
 \hline
&&&&&\\
&$\Bt$&&&$\At$&\\
&&&&&\\
\hline
\end{tabular}\cdot
\begin{tabular}{|c|}
\hline
\\
$Id_{N\times N}$\\
\\
\hline
\\
$-\At^{-1}\Bt$\\
\\
\hline
\end{tabular}=0
\]

which implies that

\[
\Mac\cdot\begin{tabular}{|c|} \hline
\\
$\mathfrak{M}_\Sbs({\bf y})$\\
\\
\hline
\\
$\mathfrak{R}_\Sbs({\bf y})$\\
\\
\hline
\end{tabular}=0,
\]

where $\mathfrak{R}_\Sbs({\bf y})=-\At^{-1}\Bt\cdot\mathfrak{M}_\Sbs({\bf
y})$.
\end{proof}

By construction, the column of $\mathfrak{M}_\Sbs({\bf y})$ indexed by
$b_{j} \in \Sbs$ corresponds to the values of $b_{j}\cdot \Lambda \in
\A^{*}$ on $b_{1},\ldots,b_{N}$. The same column in
$\mathfrak{R}_\Sbs({\bf y})$ corresponds to the values of $b_{j}\cdot
\Lambda$ on the complementary set of monomials of
$\mathrm{Mon}_{\le}(\Delta)$. The column in the stacked matrix
corresponds to the value of $b_{j}\cdot \Lambda$ on all the
monomials in $\mathrm{Mon}_{\le}(\Delta)$.  To evaluate $b_{j}\cdot
\Lambda(p)$ for a polynomial $p$ of degree $\le \Delta$, we simply
compute the inner product of the coefficient vector of $p$ with this
column.

\begin{defn}\label{SylS}
Let $\Sbs=[b_1, \ldots, b_N]$ be the basis of $\A$ as above, and let
$P\in \CC[\xb]$ be a polynomial of degree at most $D$.

Define ${\rm Syl}_{\Sbs}(P)$ to be the matrix with rows corresponding
to the coefficients of the polynomials $(b_1P),\ldots,(b_NP)$ in the
monomial basis ${\rm Mon}_{\leq}(\Delta)$ (we use here that ${\rm
deg}(b_i)\leq D$, thus ${\rm deg}(b_iP)\leq 2D\leq\Delta$).

Furthermore,  we assume that the monomials corresponding to the
columns of ${\rm Syl}_{\Sbs}(P)$ are  in the same order as the
monomials  corresponding to  the columns of $\Mac$.
\end{defn}

\begin{thm}

\[
\begin{tabular}{|ccc|}
\hline
&&\\
&${\rm Syl}_{\Sbs}(J)$&\\
&&\\
\hline
\end{tabular}\cdot
\begin{tabular}{|c|}
\hline
\\
$\mathfrak{M}_\Sbs({\bf y})$\\
\\
\hline
\\
$\mathfrak{R}_\Sbs({\bf y})$\\
\\
\hline
\end{tabular}=[Tr(b_ib_j)]_{i,j=1}^N
\]

\end{thm}

\begin{proof}
Since the $j$-th column of the matrix
$$
\begin{tabular}{|c|}
\hline
\\
$\mathfrak{M}_\Sbs({\bf y})$\\
\\
\hline
\\
$\mathfrak{R}_\Sbs({\bf y})$\\
\\
\hline
\end{tabular}
$$
represents the values of $b_j\cdot \Lambda$ on all the monomials of
degree less than or equal to $\Delta$, and the $i$-th row of ${\rm
Syl}_{\Sbs}(J)$ is the coefficient vector of $b_i J$, we have
\[
\begin{aligned}
\begin{tabular}{|ccc|}
\hline
&&\\
&${\rm Syl}_{\Sbs}(J)$&\\
&&\\
\hline
\end{tabular}\cdot
\begin{tabular}{|c|}
\hline
\\
$\mathfrak{M}_\Sbs({\bf y})$\\
\\
\hline
\\
$\mathfrak{R}_\Sbs({\bf y})$\\
\\
\hline
\end{tabular}&=\left[(b_j \cdot \Lambda)(b_i J)\right]_{i,j=1}^N\notag\\
&=\left[\Lambda(J b_i b_j)\right]_{i,j=1}^N\notag\\
&=[Tr(b_ib_j)]_{i,j=1}^N.\notag
\end{aligned}
\]
\end{proof}

We can now describe the algorithm to compute a set of multiplication
matrices $M_{x_i}$, $i=1,\ldots,m$ of the radical $\sqrt{\I}$ of
$\I$ with respect to a  basis of $\CC[\xb]/\sqrt{\I}$. To prove that
the algorithm below is correct we need the following result from
\cite[Proposition 8.3]{JaRoSza07}
which is the consequence of the  fact that the kernel of the matrix
of traces corresponds to the radical of  $\A$:

\begin{prop}
Let $\tilde{T}$ be a maximal non-singular submatrix of the matrix of
traces $T$. Let $r$ be the rank of $\tilde{T}$, and $\tilde{\Sbs}:=[b_{i_1},
\ldots, b_{i_r}]$ be the monomials corresponding to the columns of
$\tilde{T}$. Then $\,\tilde{\Sbs}$ is a basis of the algebra $\CC[{\bf
x}]/\sqrt{\I} $ and for each $k=1, \ldots, m$, the solution
$M_{x_k}$ of  the linear matrix equation
$$\tilde{T}M_{x_k}=\tilde{T}_{x_k}$$
is the  multiplication matrix of $x_k$  for $\sqrt{\I} $ with
respect to  $\tilde{\Sbs}$. Here $\tilde{T}_{x_k}$ is the $r\times r$ submatrix
of $[Tr(x_kb_ib_j)]_{i,j=1}^N$ with the same row and column indices
as in $\tilde{T}$.
\end{prop}

\noindent\begin{alg}[Radical ideal using Macaulay matrices and traces]\ \\\label{algo}
\noindent\textsc{Input}: $\f=[f_1,\ldots,f_s]\in\CC[{\bf x}]$ of
degrees $d_1,\ldots,d_s$ generating an ideal $\I$ and $\delta>0$
such
that for $k:=\delta-1$ they satisfy the conditions in Assumption  \ref{assump}. An optional input is $D\leq\delta$, which by default is set to be $\delta$.\\

\noindent\textsc{Output}: A basis $\tilde{\Sbs}$ for the factor algebra
$\CC[{\bf x}]/\sqrt{\I}$ and a set of multiplication matrices
$\{M_{x_i}|i=1,\ldots,m\}$ of $\sqrt{\I}$  with respect to the basis
$\tilde{\Sbs}$.\\

\begin{enumerate}

\item Compute ${\rm Mac}_{\Delta}(\f)$ for $\Delta:=\max(2D, \delta-1)$ as in Definition \ref{bigdelta}.

\item Compute a basis $\Sbs$ of $\mathbb{K}[{\bf x}]_{\Delta}/(\langle {\bf f}\rangle_{\Delta+1}\cap\mathbb{K}[{\bf x}]_{\Delta})$ such that the polynomials in $\Sbs$ have degrees at most $D$. Let $\Sbs=[b_1,\ldots,b_N]$.

\item Compute a random combination ${\bf y}$
of the elements of a basis of $Null(\Mac)$.

\item Compute the moment matrix $\mathfrak{M}_\Sbs({\bf y})$ defined in Definition \ref{moment} and
$\mathfrak{R}_\Sbs({\bf y})$ defined in Lemma \ref{Rlemma}.

\item Compute $\mathfrak{M}_\Sbs^{-1}({\bf y})$ and the
basis $[b_1^*,\ldots,b_N^*]$ defined in Definition \ref{dualbasis}.

\item Compute $J=\sum_{i=1}^N b_ib_i^*\;\text{ mod }  \I$ using $\Mac$.

\item Compute ${\rm Syl}_{\Sbs}(J)$ and ${\rm Syl}_{\Sbs}(x_kJ)$ for $k=1,\ldots,m$ defined in Definition \ref{SylS}.

\item Compute \\
$T=[Tr(b_ib_j)]_{i,j=1}^N=$ \begin{tabular}{|ccc|}
\hline
&&\\
&${\rm Syl}_{\Sbs}(J)$&\\
&&\\
\hline
\end{tabular} $\cdot$ \begin{tabular}{|c|}
\hline
\\
$\mathfrak{M}_\Sbs({\bf y})$\\
\\
\hline
\\
$\mathfrak{R}_\Sbs({\bf y})$\\
\\
\hline
\end{tabular}\\

and\\
$T_{x_k}$:=$[Tr(x_kb_ib_j)]_{i,j=1}^N=$ \begin{tabular}{|ccc|}
\hline
&&\\
&${\rm Syl}_{\Sbs}(x_kJ)$&\\
&&\\
\hline
\end{tabular} $\cdot$ \begin{tabular}{|c|}
\hline
\\
$\mathfrak{M}_\Sbs({\bf y})$\\
\\
\hline
\\
$\mathfrak{R}_\Sbs({\bf y})$\\
\\
\hline
\end{tabular}\quad for $k=1,\ldots,m$.

\item Compute $\tilde{T}$,  a maximal non-singular submatrix of
$T$. Let $r$ be the rank of $\tilde{T}$, and $\tilde{\Sbs}:=[b_{i_1}, \ldots,
b_{i_r}]$ be the monomials corresponding to the columns of
$\tilde{T}$.

\item For each $k=1, \ldots, m$ solve the linear matrix equation $\tilde{T}M_{x_k}=\tilde{T}_{x_k}$, where
$\tilde{T}_{x_k}$ is the submatrix of ${T}_{x_k}$ with the same row
and column indices as in $\tilde{T}$.

\end{enumerate}
\end{alg}

\begin{rem}  Since the bound given in Theorem \ref{deltatheorem} might be
   too high, it seems reasonable to design the algorithm in
   an iterative fashion, similarly to the algorithms in \cite{LaLaRo07,LaLaRo072,ReZh04}, in order to
   avoid nullspace computations for large matrices.  The bottleneck of our algorithm is doing computations with ${\rm Mac}_{\Delta}(\f)$, since its size exponentially increases as $\Delta$ increases.
\end{rem}

\begin{rem}
Note that if $s=m$ then we can use the conventional Jacobian of
$f_1, \ldots, f_m$ in the place of $J$, and any $|{\rm Mon}_\leq
(\Delta)|\times |\Sbs|$ matrix $X$ such that it has full rank and $\Mac
\cdot X={\bf 0}$ in the place of

\[\begin{tabular}{|c|} \hline
\\
$\mathfrak{M}_\Sbs({\bf y})$\\
\\
\hline
\\
$\mathfrak{R}_\Sbs({\bf y})$\\
\\
\hline\end{tabular}\,.\]
Even though this way we will not get
matrices of traces,  a system of multiplication matrices of the
radical $\sqrt{\I}$ can still be recovered:
 if $\tilde{Q}$ denotes a maximal non-singular submatrix of ${\rm Syl}_\Sbs(J)\cdot X$, and
$\tilde{Q}_{x_k}$ is the submatrix of ${\rm Syl}_\Sbs(x_k J)\cdot X$
with the same row and column indices as in $\tilde{Q}$, then the
solution $M_{x_k}$ of the linear matrix equation
$\tilde{Q}M_{x_k}=\tilde{Q}_{x_k}$ gives the same multiplication
matrix of $\sqrt{\I}$ w.r.t. the same basis $\tilde{\Sbs}$ as the above
Algorithm.
\end{rem}

\begin{rem}
As $M_{x_{k}}$ is the transpose matrix of multiplication by $x_{k}$ modulo the
radical ideal $\sqrt{\I}$, its eigenvectors are (up to a non-zero scalar) the evaluation at the roots $\zeta$
of $\I$ (see \cite{BMr98,ElMo07} for more details). The vector which 
represents this evaluation at $\zeta$ in the dual space $\A^{*}$ is
the vector of values of $[b_{1}, \ldots, b_{N}]$ at $\zeta$. To
obtain these vectors, we solve the generalized eigenvalue problem
$(\tilde{T}_{x_{k}}-z \tilde{T})\, w =0$ and compute $v =
\tilde{T}\, w$. The vectors $v$ will be of the form
$[b_{1}(\zeta),\ldots, b_{N}(\zeta)]$ for $\zeta$ a root of $\I$. If
$b_{1}=1,b_{2}=x_{1}, \ldots, b_{m+1}=x_{m}$, we can read directly
the coordinates of $\zeta$ from this vector.
\end{rem}

\subsection{The Non-Gorenstein Case}

We will now consider the case where $\A$ is not Gorenstein. The main idea of the algorithm is the same as in the Gorenstein case, except we will obtain as an output a matrix of traces with respect to an algebra $\BD$ which is a maximal Gorenstein factor of $\A$. This will still allow us to compute the multiplication matrices of the radical of $\I$ since the maximal non-singular submatrix of the trace matrix corresponding to $\BD$ is the same as that of the trace matrix of $\A$. First we will need some results to define a maximal Gorenstein factor $\BD$ of $\A$ from a random element of the nullspace of $\Mac$. 

Let $\CC$ be an arbitrary algebraically closed field. All algebras
we consider will be finite dimensional commutative $\CC$-algebras. A
local $\CC$-algebra here is an $\CC$-algebra $\BD$, with unique
maximal ideal (which we denote by $\M$) such that $\BD/\M$ is
isomorphic to $\CC$. Note that due to the fact that $\CC$ is algebraically
closed, no other residue class field is possible.

\begin{defn}
Fix $\yb\in {\rm Null}(\Mac)$ and let $\Lambda:=\Lambda_\yb\in \A^*$ defined as in Definition \ref{moment}. 
We define
$$\Rc(\Lambda):=\{ a\in A, \Lambda(ab)=0 \text{ holds for all } b\in A\}.$$

Note that $\Rc(\Lambda)=0$ iff $\Lambda(xy)$ is a nondegenerate
bilinear form on $\A$. Also, an easy calculation shows that $\Rc(\Lambda)$ is an ideal in $\A$.

Define $$\BD:=\A/\Rc(\Lambda).$$
\end{defn}

First we need the following technical lemmas.

\begin{lem} \label{ideal}$\Rc(\Lambda)$ is an ideal of $\A$, in fact, it is the largest ideal of $\A$ which
is in $\ker(\Lambda)$.

Here $\ker(\Lambda)$ is the set of elements $a\in \A$, such that
$\Lambda(a)=0$.
\end{lem}

\begin{proof} An easy calculation shows that $\Rc(\Lambda)$ is an ideal. Clearly it is
in $\ker(\Lambda)$. Conversely, if $\J$ is an ideal of $\A$ which is
in $\ker(\Lambda)$, then $\J$ is in $\Rc(\Lambda)$. Note that the sum
of ideals is an ideal again, hence there exists a unique largest
ideal of $\A$ which is  a subset of $\ker(\Lambda)$.\end{proof}

\begin{lem} (Structure Theorem on Artininan Rings, specialized to our setting,
see Atiyah-MacDonald) A $\CC$-algebra $\BD$ is the direct product of
finitely many local $\CC$-algebras.
\end{lem}

\begin{lem} \label{fact} Suppose that the $\CC$-algebra $\BD$ is the direct product of the
local $\CC$-algebras $\BD_i, (i=1,\ldots ,k)$. Then $\BD$ is
Gorenstein iff all the $\BD_i$ are Gorenstein.
\end{lem}

\begin{proof} Note that the $\BD_i$ can be viewed as ideals of $\BD$, moreover $\BD$ is the
direct sum of these (as $\CC$-subspaces).

Let $\Lambda$ be a linear function on $\BD$ such that the form
$\Lambda(xy)$ is nondegenerate. Then the resctriction $\Lambda_i$ of
$\Lambda$ on $\BD_i$ will define a nondegenerate bilinear form on
$\BD_i$. Let a be a nonzero element of $\BD_i$. Then there exists a
$b\in \BD$ such that $\Lambda(ab)\neq 0$. Write now $b$ as $b=\sum
b_j$,  with $b_j\in \BD_j$ for $(j=1,\ldots ,k)$. Note that we have
$ab_j=0$ if $j\neq  i$, hence
\[
0\not=\Lambda(ab)=\Lambda(ab_i)=\Lambda_i(ab_i),
\]
hence $\Lambda_i(xy)$ is nondegenerate.

Conversely, assume that we have linear functions
$\Lambda_i:\BD_i\rightarrow \CC $ such that the form $\Lambda_i(xy)$
is nondegenerate on $\BD_i$. We define $\Lambda$ as follows. Let
$a\in \BD$ with $a=a_1+\cdots +a_k$, where $a_i\in \BD_i$. Note that
$a$ uniquely determines the $a_i$, and the map $a\mapsto a_i$ is an
$\CC$-algebra morphism from $\BD$ to $\BD_i$. This implies that
$\Lambda(a):=\Lambda_1(a_1)+\cdots +\Lambda_k(a_k)$ is a correct
definition and $\Lambda$ is a linear function on $\BD$. Moreover, it
is easily seen that $\Lambda(xy)$ is nondegenerate.
\end{proof}

\begin{lem}\label{fact2} $\Lambda$ induces a linear function $\Lambda'$ on the factor $\BD=\A/\Rc(\Lambda)$. For the function $\Lambda'$ on $\BD$ we have that $\Rc( \Lambda')=0$, hence $\BD$ is Gorenstein.
\end{lem}

\begin{proof} We set $\Lambda'(a+\Rc(\Lambda))=\Lambda(a)$. It is routine to check that $\Lambda'$ is  a
correctly defined linear function on $\BD$. Suppose that
$a+\Rc(\Lambda)$ is in $\Rc(\Lambda')$. Then $\Lambda(ac)=0$ holds for
every $c\in \A$, hence $a\in \Rc( \Lambda)$ and therefore
$a+\Rc(\Lambda)=0$ in $\BD$.\end{proof}

\begin{lem} Every Gorenstein factor $\BD$ of an $\CC$-algebra $\A$ can be obtained via
a linear function $\Lambda$ on $\A$, as outlined by Lemma
\ref{fact2}.
\end{lem}

\begin{proof} Let $\BD=\A/\J$ be a Gorenstein factor of $\A$. Let $\Lambda'$ be a linear function
on $\BD$ with $\Rc(\Lambda')=0$. We define $\Lambda$ on $\A$ as
follows. Let $\Lambda(a):=\Lambda'(a+\J)$ for $a\in \A$. Clearly
$\Lambda$ will be a linear function on $\A$. Let $a\in \A$ such that
$\Lambda(ab)=0$ for every $b\in \A$. Then
$0=\Lambda(ab)=\Lambda'(ab+\J)=\Lambda'((a+\J)(b+\J))$, giving that
$a+\J\in \Rc(\Lambda')$, hence $a\in \J$. This shows that
$\Rc(\Lambda)\subseteq \J$. The reverse containment is immediate,
therefore $\Rc(\Lambda)=\J$. Now one can directly check that
$\Lambda'$ is obtained from $\Lambda$ via the factor construction of
Lemma \ref{fact2}.
\end{proof}

\begin{lem} \label{local} Let $\A$ be an $\CC$-algebra. Then the maximal dimensional Gorenstein
factors of $\A$ are obtained as follows: from every direct factor
$\A_i$ (this is local) factor out an ideal $\J_i$ so that
$\A_i/\J_i$ is a maximal Gorenstein factor of $\A_i$.
\end{lem}

\begin{proof}  Any ideal $\J$ of $\A$ is the direct sum of the ideals $\J_j=\A_j\cap \J$.
On the other hand $\A/\J$ is Gorenstein iff the local factors
$\A_i/\J_i$ are (Lemma \ref{fact}).\end{proof}

The following Theorem shows that we can get maximal Gorenstein factors of $\A$ from random linear forms on $\A$ with high probability, similarly as in the Gorenstein case. 

\begin{thm} The maximal Gorenstein factors of $\A$ can be obtained with high probability as $\BD:=\A/\Rc(\Lambda)$, where
$\Lambda$ is a random linear function on $\A$. 
\end{thm}

\begin{proof} Let $a_1, \ldots , a_m$ be a basis of $\A$ over $\CC$. Then for a linear
function $\Lambda: \A\rightarrow \CC$ the dimension of
$\A/\Rc(\Lambda)$ is the rank of the matrix $[\Lambda(a_ia_j)]$. Thus,
maximal Gorenstein factors are obtained if the rank of the matrix is
maximal (as $\Lambda$ ranges over the $\CC$-dual of $\A$). If one
fixes a dual basis of $\A$, and writes $\Lambda$ as a linear
combination of these basis functions, then, the entries of the
matrix $[\Lambda(a_ia_j)]$ will be linear polynomials of the
coordinates $\gamma_1, \ldots, \gamma_m$ of $\Lambda$. Now consider
a $\Lambda$ which achieves the maximal rank $k$, and consider a
corresponding $k \times k$ minor of the matrix that has a nonzero
determinant. This determinant is not identically zero, as a function
of the $\gamma_j$, hence it will be nonzero on a Zariski open set.
The linear functions corresponding to the points in this set will
define maximal Gorenstein factors.
\end{proof}

We now show that any maximal Gorenstein factor will allow us to compute the radical of $\A$.

\begin{thm}\label{rad}
Assume that $\Lambda$ is such that the corresponding bilinear form on $\A$ has maximal rank. Then $$\Rc(\Lambda)\subseteq {\rm Rad}(\A).$$
\end{thm}

\begin{proof}
By Lemma \ref{local}, we can assume that $\A$ is local. By Lemma \ref{ideal}, $\Rc(\Lambda)$ is an ideal. Since $\A/{\rm Rad}(\A)\cong \CC$ is Gorenstein and $\BD$ is the maximal Gorenstein factor of $\A$ we have that $\Rc(\Lambda)\neq\A$. Therefore $\Rc(\Lambda)$ is a subset of the unique maximal ideal $\M={\rm Rad}(\A)$.
\end{proof}

Using the previous results, we are now ready to define the main ingredients of our algorithm in the non-Gorenstein case, which are analogous to the Gorenstein case except that instead of working in $\A$ we are going to work in $\BD$. We will use the notation of the previous subsection.

We can obtain a basis for $\BD$ as follows. Let $\Sbs=[b_1,\ldots,b_N]$ be a basis for $\A$ and $\yb\in {\rm Null}(\Mac)$ such that the moment matrix  $\mathfrak{M}_{\Sbs}(\yb)$ has maximal rank. Since the columns of $\mathfrak{M}_{\Sbs}(\yb)$ correspond to $\Sbs$, taking a maximal nonsingular minor of this matrix will define a subset  $\Sbs_{\alpha}=[b_{\alpha_1},\ldots,b_{\alpha_r}]$ of $\Sbs$ corresponding to the columns of this submatrix. Then $\Sbs_{\alpha}$ will form a basis for $\BD$ as we prove in the following proposition.

\begin{prop}
 $\Sbs_{\alpha}=[b_{\alpha_1},\ldots,b_{\alpha_r}]$ forms a basis for $\BD$.
\end{prop}

\begin{proof}
Consider the moment matrix  $\mathfrak{M}_{\Sbs}(\yb)$ with columns corresponding to $\Sbs$ and let $r$ be its rank. Since $\Sbs_{\alpha}$ corresponds to a set of basic columns of $\mathfrak{M}_{\Sbs}(\yb)$, there exists a basis ${\bf v}_1,\ldots,{\bf v}_{N-r}$ for ${\rm Null}(\mathfrak{M}_{\Sbs}(\yb))$ which can be extended to a basis of $\CC^N$ by adding the unit vectors ${\bf e}_{\alpha_i}:=[\delta_{\alpha_i,j}]_{j=1}^N$ for $i=1,\ldots,r$. 

Let $v_i$ be the element of $\A$ obtained by taking the linear combination of
$b_1,\ldots,b_N$ corresponding to the coordinates of ${\bf v}_i$ for
$i=1,\ldots,N-r$. Then  it is easy to see that $v_i\in \Rc(\Lambda)$. Thus
the elements of $\Sbs_{\alpha}$ correspond to a basis of $\A/\Rc(\Lambda)=\BD$.
\end{proof}

Here we need to define the moment matrix.

\begin{defn}
Let $\Sbs_{\alpha}$ be defined as above. Define
$$
J_{\alpha}:=\sum_{i=1}^kb_{\alpha_i}b_{\alpha_i}^*
$$
similarly as in \ref{J}.
\end{defn}

\begin{thm}
$[Tr(b_{\alpha_i}b_{\alpha_j})]_{i,j=1}^k=$ \begin{tabular}{|ccc|}
\hline
&&\\
&${\rm Syl}_{\Sbs_{\alpha}}(J_{\alpha})$&\\
&&\\
\hline
\end{tabular} $\cdot$ \begin{tabular}{|c|}
\hline
\\
$\mathfrak{M}_{\Sbs_{\alpha}}(\yb)$\\
\\
\hline
\\
$\mathfrak{R}_{\Sbs_{\alpha}}(\yb)$\\
\\
\hline
\end{tabular}.
\end{thm}

Using the following theorem we get that the maximal nonsingular minor of the smaller trace matrix $[Tr(b_{\alpha_i}b_{\alpha_j})]_{i,j=1}^k$ suffices to compute the radical.

\begin{thm}
Let $\Sbs=[b_1,\dots,b_N]$ be a basis for $\A$ and $\Sbs_{\alpha}=[b_{\alpha_1},\ldots,b_{\alpha_r}]$ a basis for $\BD$, where $\alpha_1<\cdots< \alpha_r$ are in $\{1,\ldots,N\}$.  Then
$$
{\rm rank}[Tr(b_{\alpha_i}b_{\alpha_j})]_{i,j=1}^r={\rm rank}[Tr(b_ib_j)]_{i,j=1}^N.
$$ 
(As before, by a slight abuse of notation, we use the same notation for elements in $\A$, in $\BD$, and their common preimages in $\CC[\xb]$.)
\end{thm}
\begin{proof}
This follows from Theorem \ref{rad} 
and the fact that the rank of the trace matrix is  $\dim\A/{\rm Rad}(\A)=\dim\BD/{\rm Rad}(\BD)$.
\end{proof}

\section{Affine Complete Intersection Ideals}
\subsection{Univariate Case}\label{unisection}

In this section we will follow the work of Mourrain and Pan
\cite{MouPa00}. We start by defining the univariate Bezout matrix.

\begin{defn}
Let $f,g\in \KK[x]$ be two univariate polynomials such that
$\deg g\leq \deg f=d$, and let $y$ be a new variable. Then the
Bezoutian $\Theta_{f,g}$ of $f$ and $g$
 is the polynomial
\[
\BEZ_{f,g}(x,y)=\frac{f(x)g(y)-f(y)g(x)}{x-y}=\sum_{0\leq i,j\leq
d-1}c_{ij}x^iy^j.
\]

The Bezout Matrix $B_{f,g}$ of  $f$ and $g$  is the $d\times d$
matrix
\[
[B_{f,g}]_{ij}=c_{ij}.
\]
\end{defn}

We will need the following definition:
\begin{defn}
The \emph{Horner basis} for the polynomial $f$ is the set
$\{H_{d-1},\ldots,H_0\}$ with
\[
H_i(x)=a_{i+1}+\cdots+a_dx^{d-i-1}\quad\text{ for } i=0,\ldots, d-1.
\]
Note that in terms of the Horner polynomials, we have that
\[
\BEZ_{f,1}(x,y)=\sum_{i=0}^{d-1}x^iH_i(y).
\]
\end{defn}
The following theorem connects the Bezoutian of $f$ and its
derivative $f'$ with the matrix of traces of $f$ with respect to the
Horner basis.

\begin{thm}\label{main}
\[
B_{f,f'}=[Tr(H_iH_j)]_{i,j=0}^{d-1},
\]
where $[Tr(H_iH_j)]_{i,j=0}^{d-1}$ is the matrix of traces of $f$ in
the Horner basis (see eg. \cite{Card95}, \cite{MouPa00}).
\end{thm}

Theorem \ref{main} implies that using Dickson's Lemma one can
compute the square-free factor of $f$ by simply computing the kernel
of $B_{f,f'}$. It's natural to ask how our method based on Dickson's
lemma relates to computing the square-free factor of $f$ via
computing $\frac{f}{gcd(f,f')}$. The following proposition shows
that computing $f/\gcd(f,f')$ to get the square-free factor using
the Bezout matrix is computationally equivalent to using Dickson's
Lemma.

\begin{prop}
The smallest degree polynomial of the form
$\sum_{i=0}^{d-1}r_iH_i(x)$ such that $[r_0, \ldots, r_{d-1}]^T$ is
in the kernel of $B_{f,f'}$, is equal to $f/\gcd(f,f')$.
\end{prop}

\subsection{Multivariate Case}\label{multisection}

For the multivariate case we will first define the multivariate
analogue of the Bezout matrix (also referred to as Dixon matrix in
the literature). The papers \cite{carmou96} and \cite{KaSaYa} are
good references for the Bezout (Dixon) matrix described below.

\begin{defn}\label{bezmatrixdef}
Let
$${\bf f}:=[f_1, \ldots, f_m]\in \C[x_1, \ldots, x_m]^m
$$
and consider an additional polynomial $f_0\in \C[x_1, \ldots, x_m]$.
We use the notation ${\bf x}=[x_1, \ldots ,x_{m}]$, ${\bf y} = [y_1,
\ldots , y_{m}]$ and $$ X_{0} = [x_1, \ldots , x_{m}], X_{1}
= [y_1,x_2, \ldots , x_{m}], \ldots,X_{m} =[y_1,y_2, \ldots ,
y_{m}].
$$
The \emph{Bezoutian}  of the system  $[f_0, f_1, \ldots, f_n]$,
denoted by $\BEZ_{{f_0}}$, is a polynomial in the variables ${\bf
x}$ and ${\bf y}$ defined as follows: 
\begin{eqnarray}\label{Bezoutian} \BEZ_{f_0}({\bf x},
{\bf y}) := \det\left| {\begin{array}{cccc} { f_0}(X_0) &
{\displaystyle \frac {{ f_0}(X_0) - { f_0}(X_1)}{x_1 - {y_{1}}}} &
\ldots & {\displaystyle \frac {{ f_0}(X_{m-1}) - { f_0}(X_{m})}{x_{m} - {y_{m}}}} \\
{ f_1}(X_1) & {\displaystyle \frac {{f_1}(X_0) - {f_1}(X_1)}{x_1 -
{y_{1}}}} & &{\displaystyle \frac {{ f_1}(X_{m-1}) - {
f_1}(X_{m})}{x_{m} - {y_{m}}}} \\ [2ex] \vdots & \vdots & &
\vdots\\{ f_m}(X_1) & {\displaystyle \frac {{ f_m}(X_0) - {
f_m}(X_1)}{x_1 - {y_{1}}}} & \ldots &{\displaystyle \frac {{
f_m}(X_{m-1}) - { f_m}(X_{m})}{x_{m} - {y_{m}}}}
\end{array}} \right| .
\end{eqnarray} 
 The \emph{Bezout matrix} of the system $[f_0, f_1, \ldots, f_m]$,
denoted by $\BD_{{f_0}}$, is the coefficient matrix of the
Bezoutian, i.e. if we write
$$
\BEZ_{f_0}({\bf x},{\bf y}) = \sum_{\alpha \in E, \;\; \beta\in E'}
c_{\alpha, \beta}(f_{0})\; {\bf x}^{\alpha}{ \bf y}^{ \beta}
$$
where $E$ and $E'$ are subsets of $\N^{m}$ and
$c_{\alpha,\beta}(f_{0})\in\C$, then the Bezout matrix of $[f_0,f_1
\ldots,f_m]$ is the $|E'|\times |E|$ matrix
$$ \BD_{f_0} :=
\left( c_{\beta, \alpha}(f_{0})\right)_{\beta\in E', \;\;\alpha \in
E }.
$$
\end{defn}
We denote by $\BEZ^{\xb}_{f_{0}}$ the map
\begin{equation}\label{defbezx}
\BEZ^{\xb}_{f_{0}}: (\lambda_{\beta})_{\beta\in E'} \mapsto
\sum_{\alpha \in E, \;\; \beta\in E'} c_{\alpha, \beta}(f_{0})\;
{\bf x}^{\alpha} \lambda_{\beta}
\end{equation}
and by  $\BEZ^{\yb}_{f_{0}}$ the map
\begin{equation}\label{defbezy}
\BEZ^{\yb}_{f_{0}}: (\lambda_{\alpha})_{\alpha\in E} \mapsto
\sum_{\alpha \in E, \;\; \beta\in E'} c_{\alpha, \beta}(f_{0})\;
\lambda_{\alpha} { \bf y}^{ \beta}.
\end{equation}

Our goal is to compute a matrix of traces for the system ${\bf
f}=[f_1, \ldots,f_m]$ from the Bezout matrix $\BD_J$ of the system
$[J,f_1, \ldots, f_m]$ analogously to the univariate case, where $J$
is the Jacobian of $f_1, \ldots,f_m$. As we mentioned in the
introduction, in general the Bezout matrix  $\BD_J$ is not a matrix
of traces, which can be easily seen by comparing sizes. However, to
obtain a matrix of traces of ${\bf f}$ one can define a reduced
version of the Bezoutian and the Bezout matrix as follows.

\begin{defn}\label{redbezout}
Let ${\bf f}=[f_1, \ldots, f_m]\in \C[x_1, \ldots, x_m]^m$ and
assume that the factor algebra $\A({\bf x}):=\C[{\bf x}]/\I$ has
dimension $N$ over $\C$, where $I$ is the ideal generated by the
polynomials in ${\bf f}$. For some $f_0\in \C[{\bf x}]$ let
$\BEZ_{f_0}({\bf x}, {\bf y})$ be the Bezoutian of the system $f_0,
f_1, \ldots, f_m$. Let  $B=[b_1, \ldots,b_N]$ and $B'=[b_1',\ldots
b_N']$ be bases for $A({\bf x})$ and $A({\bf y})$, respectively.
Then we can uniquely write
$$
\BEZ_{f_0}({\bf x}, {\bf y})=\sum_{b\in B, b'\in
B'}\beta_{b,b'}(f_{0})\,b({\bf x})b'({\bf y}) + F({\bf x}, {\bf y})
$$
where $F({\bf x}, {\bf y})\in (I({\bf x}),I({\bf y}))$ and
$\beta_{b,b'}(f_{0})\in\C$. We define  the \emph{reduced Bezoutian}
$\RBEZ_{f_0}$ with respect to the bases $B$ and $B'$  as
$$\RBEZ_{f_0}:= \sum_{b\in B,
b'\in B'}\beta_{b,b'}(f_{0})\,b({\bf x})b'({\bf y})
$$
and the \emph{reduced Bezout matrix} $\RBD_{f_0}$ with respect to
the bases $B$ and $B'$ to be the $N\times N$ matrix
$$
\RBD^{B,B'}_{f_0}:=\left(\beta_{b,b'}(f_{0})\right)_{b\in B, b'\in
B'}.
$$
\end{defn}
We are going to use the following theorem \cite{MEBMissac99}:
\begin{thm}\label{thmdiag}
There exists (dual) bases $\Theta:=[\theta_{1}(\xb), \ldots,
\theta_N(\xb)]$ and  $\Omega:=[\omega_{1}(\yb),
\ldots,\omega_{N}(\yb)] $ of $A(\xb)$ and $A(\yb)$ such that for all
polynomial $f \in \CC[\xb]$, we have
$$
\BEZ_{f}:=
\sum_{i,j}\beta_{i,j}(f)\,\theta_{i}(\xb)\,\omega_{j}(\yb) +
F(\xb,\yb),
$$
with $F(\xb,\yb)\in I(\xb)\otimes I(\yb)$ and such that
$$
\RBEZ_{1} = \sum_{i} \theta_{i}(\xb)\,\omega_{i}(\yb).
$$
\end{thm}
In \cite{MouPa00} the following expression was given for the reduced
Bezout matrix $\RBD_J$  of the system $[J,f_1, \ldots, f_m]$ in
terms of a matrix of traces  of $f_1, \ldots,f_m$:
\begin{thm} \label{redbeztrace}  Let ${\bf f}$, $A({\bf x})$, $A({\bf y})$ and the
bases $\Theta=(\theta_{i}),\; \Omega=(\omega_{i})$ be as in
Definition \ref{redbezout} and Theorem \ref{thmdiag}. Let $J$ be the
Jacobian of $f_1, \ldots,f_m$, and consider the reduced Bezout
matrix $\RBEZ_{J}$ with respect to the bases $\Theta$ and $\Theta$.
Then
$$
\RBD^{\Theta,\Theta}_{J}=[Tr(\omega_i\omega_j)]_{i,j=1}^{N}.
$$
\end{thm}
Using the relation
$\RBD^{\Theta,\Theta}_{J}=\RBD^{\Theta,\Omega}_{J}\,
\RBD^{\Theta,\Theta}_{1}$, we deduce that
$$
\RBD^{\Theta,\Omega}_{J}=[Tr(\omega_i\,\theta_j)]_{i,j=1}^{N},
$$
so that $\lambda=[\lambda_{i}]\in \ker \RBD^{\Theta,\Omega}_{J}$ iff
$$
Tr(\omega_{i}\, (\sum_{j=1}^{N} \lambda_{j}\,\theta_{j}))=0,\
i=1,\ldots,N,
$$
or equivalently iff
$$
r := \sum_{j=1}^{N} \lambda_{j}\,\theta_{j}(\xb)
  \in \sqrt{I}.
$$
Because of the block diagonal form of the Bezoutian matrices in a
common basis (Theorem \ref{thmdiag}), we deduce that if $\Lambda$ is
an element of $\ker(\BEZ_{J}) $ then
$$
\BEZ_{1}^{\xb}(\Lambda) = r(\xb) + h(\xb),
$$
where $h\in I$, $\BEZ_{1}$ is the Bezoutian matrix of $1$ in the (monomial)
bases $(\xb^{\alpha})_{\alpha\in E}$, $(\yb_{\beta})_{\beta\in E'}$,
$\BEZ_{1}^{\xb}$ is the corresponding map defined in \eqref{defbezx},
and $\im( \BEZ^{\yb}_{1})$ is the space generated by the coefficient
vectors with respect to $(\yb^{\beta})_{\beta\in E'}$ of the
polynomials in the image of the map $\BEZ^{\yb}_{1}$ (see
\eqref{defbezy}). Then, we have the following theorem:
\begin{thm}\label{radbez}
Using the previous notation we have that
$$
\sqrt{I} = \BEZ_{1}^{\xb}(\ker (\BEZ_{J}^{\xb}))+ I(\xb).
$$
\end{thm}
\begin{proof}
Because of the block diagonal form of the Bezoutian matrices in a common basis
(Theorem \ref{thmdiag}) and the previous discussion, we deduce that if $\Lambda$ is an element of
$\ker (\BEZ_{J})$ then
$$
\BEZ_{J}^{\xb}(\Lambda) = 
\sum_{i,j}\beta_{i,j}(J)\,\theta_{i}(\xb)\,\Lambda(\omega_{j}(\yb)) +
F^{\Lambda}(\xb) =0
$$
where $F^{\Lambda}(\xb) \in I$. By the previous discussions, 
$$ 
\sum_{j}\,\theta_{j}(\xb)\,\Lambda(\omega_{j}(\yb)) = 
\BEZ_{1}^{\xb}(\Lambda) \in \sqrt{I}.
$$
\end{proof}

Note that the role of $\xb$ and $\yb$ can be exchanged in this
theorem. Note also that $\ker(\BEZ_{J}^{\xb})$ can be replaced by $\ker(\BEZ_{J}^{\xb}) \cap
\ker(\BEZ_{1}^{\xb})^{\bot}$ in this theorem.


 A question that remains is how to
compute the  multiplication matrices ${M}_{x_1},
\ldots,{M}_{x_m}$ of the radical ${\sqrt{I}}$.

In order to compute the reduced Bezout matrix $\RBD_{f_0}$ of the
system $[f_1,\ldots,f_m]$ with respect to some bases
$[a_1,\ldots,a_r]$ of $A({\bf x})$ and $[b_1, \ldots,b_r]$
of $A({\bf y})$, it is sufficient to find expressions of the form
\[
\begin{aligned}
{\bf x}^{\alpha}&=\sum_{b\in B}c_b b({\bf x})+F({\bf
x}),\;\text{for all  }\alpha\in E\quad\text{and}\notag\\
{\bf y}^{\beta}&=\sum_{b'\in B'}c_{b'} b'({\bf y})+G({\bf
y}),\;\text{for all }\beta\in E'\notag
\end{aligned}
\]
where $F,G\in I$, $c_b,c_{b'}\in\C$ and $E$ and $E'$ were defined
in Definition \ref{bezmatrixdef}. Define $$ V:=\langle{\bf
x}^\alpha|\alpha\in E\rangle \quad\text{ and } \quad
W:=\langle{\bf y}^\beta|\beta\in E'\rangle. $$
Assuming that
$b_i\in V$ and $b'_i\in W$ for $i=1,\ldots,l$, the task is to find
enough linear combinations of the monomials corresponding to the
rows and the columns of the Bezout matrices which belong to the
ideal $I$.

We follow the approach described in \cite{mb-bqrs-jsc-05}, where it
was shown that the computation of the Bezoutians $\BEZ_{x_i}$ of the
system $[x_i,f_1,\ldots,f_m]$, for $i=1,\ldots,m$, as well as the
Bezoutian $\BEZ_{1}$ of the system $[1,f_1,\ldots,f_m]$, gives
sufficient information of the structure of $I$ in order to find the
reduced Bezout matrix $\RBD_{f_0}$ for any $f_0\in\C[{\bf x}]$. In
order to get the structure of ${\sqrt{I}}$ we simply have
to add the polynomials in $\xb$ (resp. $\yb$), obtained from Theorem
\ref{radbez}. 

Here we describe a summary of this method.
First notice that
\[
x_i\BEZ_1-\BEZ_{x_i}\in I({\bf x})\quad\text{and}\quad
y_i\BEZ_1-\BEZ_{y_i}\in I({\bf y}) \quad i=1, \ldots, m.
\]

 The {\bf initial step} of the method  is to obtain ideal elements in $\sqrt{I}(\bf x)$ (resp. $\sqrt{I}(\bf x)$) which
are in $V$ (resp. $W$)   from $x_i\BEZ_1-\BEZ_{x_i}$ 
(resp. $y_i\BEZ_1-\BEZ_{y_i}$) and also from $\BEZ^{\xb}_{1}\, (\ker(\BEZ_{J}^{\xb})
)$ ( resp. $\BEZ^{\yb}_{1}\,(\ker(\BEZ_{J}^{\yb}))$). The elements in $\sqrt{I}(\bf x)\cap V$ obtained by the initial step are denote $K_0$, and the ones in 
$\sqrt{I}(\bf y)\cap W$ are denoted by $H_0$. \\

For any vector space $K \subset R$, we denote by $K^+$, the vector space $K^+
= K + x_1 K + \cdots + x_m K.$ The notation $K^{[ n ]}$ means $n$ iterations
of the operator $+$, starting from $K$.

To prove that we get the quotient structure by the radical ideal $\sqrt{I}$,
we will assume that $V$ is connected to $1$, that is, 
$V$ contains $1$ and for any $v \in V-\langle 1 \rangle$, there exists $l>0$
such that $v \in{\rm span}(1)^{[l ]}$ and $v= v_{0} + \sum_{i=1}^{m} x_{i} v_{i}$ with
$v_{i} \in {\rm span}(1)^{[l-1]}\cap V$ for $i=0,\ldots,m$.

In order to obtain additional ideal elements, the following steps are 
used \cite{mb-bqrs-jsc-05}:
\begin{description}
\item[Saturation step] Finds new ideal elements by multiplying
the already computed ideal elements by the variables $x_i$ for all
$i=1,\ldots,m$.

\item[Column reduction step] Finds new bases for the vector
spaces $V$ and $W$ such that the new basis for $V$ contains
previously computed elements in ${\sqrt{I}}({\bf x})\cap
V$, and also that the Bezout matrix $\BD_1$, written in terms of
these new bases, has a lower block triangular structure. By writing
the matrices $\BD_{x_i}$ in terms of the new bases for $V$ and $W$,
one can obtain new elements in ${\sqrt{I}}({\bf x})\cap
V$.

\item[Diagonalization step] After the column reduction step one
can transform $\BD_1$ into a block diagonal form which, by repeating
the same transformation on the matrices $\BD_{x_i}$, can possibly
reveal new ideal elements.

\item[Row reduction step] Same as the column reduction step, with
the roles of ${\bf x}$ and ${\bf y}$ interchanged.
\end{description}

They are used in the following iterative algorithm:
\noindent\begin{alg}[Radical of an affine complete intersection]\ \\\label{algobez}
\noindent\textsc{Input}: $\f=[f_1,\ldots,f_m]\in\CC[{\bf x}]$ generating an ideal $I$, which have a
finit number of complex roots.\\
\noindent\textsc{Output}: $M_{x_1}, \ldots, M_{x_m}$ a system of
multiplication matrices for the radical ideal $\sqrt{I}$. 
\begin{itemize}
 \item Compute the Bezoutian matrices $\BD_{1},\BD_{x_1}, \ldots, \BD_{x_m}$ of
$1,x_{1},\ldots,x_{m}$ and $f_{1},\ldots,f_{m}$ and  $\BD_{J}$.
 \item Using the initial step, define  $K := {K}_{0}$; $H:= {H}_{0}$;
 {\bf notsat} \texttt{:= true}.
 \item While {\bf notsat}
\begin{itemize}
      \item Apply the saturation step on ${K}$ and $H$;

      \item Apply the column reduction step;

      \item Apply the diagonalisation step;

      \item Apply the row reduction step;

      \item If this extends strictly ${K}$;  or ${H}$, then \\
\ \ \  let {\bf
notsat} \texttt{:= true}, otherwise let {\bf notsat} \texttt{:= false}.
\end{itemize}
\item Return \begin{equation}\label{M'}
M_{x_i} := N_{1}^{-1} N_{x_i}, \ i= 1,\ldots, m
\end{equation}
where $N_{x_i}$ is the matrix reduced from $\BD_{x_i}$, at the end of the loop.

\end{itemize}
\end{alg}
The loop terminantes because the size of the matrices are
decreasing. If it ends with matrices of non-zero size, $N_{1}$  is necessarily
invertible. At the initial step and all along the computation, we have
$K\subset \sqrt{I}({\bf x}), H\subset \sqrt{I}({\bf y})$, since the different steps 
are valid  modulo $I \subset \sqrt{I}$.
We denote by $[a_{1}, \ldots, a_{r}]$ 
(resp.  $[b_{1}, \ldots, b_{r}]$) the linearly independent polynomials
indexing the rows (resp. columns) of $M_{i}$ and $A$ (resp. $B$) the vector
space they span. By construction, the vector space $V$ (resp. $W$) decomposes
as $V= A + {\rm span}( K )$ (resp. $W= B+ {\rm span}( H )$) where $K$ and
$H$ are the sets of relations in $\sqrt{I}$ updated in the reduction
steps during the algorithm. We complete $a_{i}, i=1,\ldots, r$ 
(resp. $b_{i}, i=1,\ldots, r$)
in a basis
$a_{1},\ldots, a_{|E|}$ of $V$ (resp. $b_{1},\ldots, b_{|E'|}$ of $W$) with 
$a_{i}\in K$  (resp. $b_{i}\in H$)  for $i>r$.

We have the following theorem, which allow us to compute the radical of 
an affine complete inetrsection, based on simple  algebra tools:
\begin{thm}
Let $ f_1, \ldots, f_m$ be as above. Upon termination,
Algorithm \ref{algobez}  computes new bases $[a_1, \ldots,
a_{|E|}]$ for $V$ and $[b_1, \ldots, b_{|E'|}]$ for $W$,
 such that 
 \begin{itemize}
 \item $[a_1, \ldots, a_{r}]$ (resp. $[b_1, \ldots, b_{r}]$) is a basis of
$\C[\xb]/{\sqrt{I}}(\xb)$
(resp. $\C[\yb]/{\sqrt{I}}(\yb)$),
 
 \item the ouput matrices $M_{i}$  are the (resp. transpose) matrix of
multiplication of $x_i$ with respect to the basis $[a_1, \ldots, a_{r}]$
(resp. $[b_1, \ldots, b_{r}]$) for $i=1,
\ldots m$,

 \item 
$a_{r+1},\ldots, a_{|E|}\in \sqrt{I}(\xb)\cap V \quad
\text{ and } \quad b_{r+1},\ldots, b_{|E'|}\in
{\sqrt{I}}(\yb)\cap W.$
 \end{itemize}
\end{thm}
\begin{proof}
The proofs of \cite{mb-bqrs-jsc-05}[Lemma 5.3, Proposition 5.3-5.6] apply
also here in order to show that the matrices $M_{x_i}$ commute.
Moreover, by construction,  $K \subset V\cap\sqrt{I}({\bf x})$ such that $V = A \oplus{\rm span}( K)$, and  $K$ satisfies the
following relations: $f_{i}\in \langle K\rangle$ for $i=1,\ldots, m$ (same proof as in
\cite{mb-bqrs-jsc-05} [Proposition 5.7]), 
$\langle K\rangle \subset \sqrt{I}$ by definition of the reduction steps, 
$B_1^{\xb}(\ker B_{J}^{\xb})\subset \langle K_{0}\rangle \subset \langle K\rangle$. Therefore, by
Theorem \ref{radbez}, this shows that $\langle K \rangle=\sqrt{I}$. 
As in \cite{mb-bqrs-jsc-05} [Theorem 5.8],  we can assume that 
$1\in A$ and that we have an exact sequence:
\begin{eqnarray*}
0 \rightarrow \langle K\rangle  \rightarrow \CC[\xb] & \rightarrow & A \rightarrow 0\\
 f & \mapsto & f(\mathbf{M})(1) 
\end{eqnarray*}
where $f(\mathbf{M}):=f(M_{x_1}, \ldots ,M_{x_m})$ is the linear operator of $A$, obtain by replacing the variable
$x_{i}$ by $M_{x_i}$.
As $\langle K\rangle =\sqrt{I}$, this shows that $A \sim  \CC[\xb]/\sqrt{I}$ and the basis of $A$ is
a basis in $ \CC[\xb]/\sqrt{I}$.
\end{proof}

\section{Conclusion}

In an earlier work  we gave an algorithm to compute matrices of traces and
the radical of an ideal $\I$ which has finitely many projective
common roots, none of them at infinity and its factor algebra is
Gorenstein. The present paper considers an extension of the
above algorithm which also works in the non-Gorenstein case and for
systems which have roots at infinity, as well as an alternative
method using Bezout matrices for the affine complete intersection
case to compute the radical $\sqrt{\I}$.

\bibliographystyle{plain}

\end{document}